\documentclass[12pt, draftclsnofoot, onecolumn]{IEEEtran}
\usepackage[utf8]{inputenc}
\usepackage{tabularx}
\usepackage{amsmath}
\usepackage{amsthm}
\usepackage{amssymb}
\usepackage{enumerate}
\usepackage{latexsym}
\usepackage{mathtools}
\usepackage{multirow}
\usepackage[ruled,linesnumbered]{algorithm2e}
\usepackage{longtable}

\SetKwRepeat{Do}{do}{while}

\usepackage[unicode=true,
bookmarks=true,bookmarksnumbered=true,bookmarksopen=true,bookmarksopenlevel=1,
breaklinks=false,pdfborder={0 0 0},pdfborderstyle={},backref=false,colorlinks=false]
{hyperref}

\makeatletter
\theoremstyle{plain}

\newtheorem{theorem}{Theorem}[section]
\newtheorem{lemma}[theorem]{Lemma}

\newtheorem{remark}[theorem]{Remark}
\newtheorem{conj}[theorem]{Conjecture}
\newtheorem{problem}[theorem]{Open Problem}

\newtheorem{example}[theorem]{Example}

\newcommand{\ord}{{\mathrm{ord}}}

\newcommand{\lcm}{{\mathrm{lcm}}}
\newcommand{\tr}{{\mathrm{Tr}}}

\newcommand{\gf}{{\mathrm{GF}}}

\newcommand{\Z}{\mathbb{{Z}}}

\newcommand{\m}{\mathsf{M}}

\newcommand{\C}{{\mathcal{C}}}
\newcommand{\M}{{\mathsf{M}}}

\newcommand{\bc}{{\mathbf{c}}}

\newcommand{\bone}{{\mathbf{1}}}

\usepackage{color, colortbl}
\usepackage[noadjust,sort,compress]{cite}

\@ifundefined{showcaptionsetup}{}{%
	\PassOptionsToPackage{caption=false}{subfig}}
\usepackage{subfig}

\definecolor{Gray}{gray}{0.9}
\usepackage[first=0,last=9]{lcg}

\begin{document}

\title{Infinitely many families of distance-optimal binary linear codes with respect to the sphere packing bound}

\author{Hao Chen, Conghui Xie, Cunsheng Ding
        \thanks{
        Hao Chen is with the College of Information Science and Technology, Jinan University, Guangzhou, Guangdong, 510632, China (haochen@jnu.edu.cn). 
        Conghui Xie is with the Hetao Institute of Mathematics and Interdisciplinary Sciences,  Shenzhen, Guangdong, 518033, China (conghui@stu2021.jnu.edu.cn).
        Cunsheng Ding is with the Department of Computer Science and Engineering, Hong Kong University of Science and Technology, Hong Kong, China (cding@ust.hk).

        The research of Hao Chen was supported by the NSFC Grant 62032009. 
        The research of Cunsheng Ding was supported by the Hong Kong Research Grants Council under Project 16301123.}
}

\maketitle

\begin{abstract}
R. W. Hamming published the Hamming codes and the sphere packing bound in 1950.  
In the past 75 years,  infinite families of distance-optimal linear codes over finite fields with minimum distance at most 8 with respect to the sphere packing bound have been reported in the literature.  However, 
it is a 75-year-old open problem in coding theory whether there is an infinite family of distance-optimal linear codes over finite fields with arbitrarily large minimum distance with respect to the sphere packing bound.  
This main objective of this paper is to settle this long-standing open problem in coding theory. 

As by-products,  several infinite families of distance-optimal binary codes with small minimum distances are presented.  Two infinite families of binary five-weight codes are reported.  Some open problems are also proposed. 

\end{abstract}


\section{Introduction}

\subsection{Linear codes, cyclic codes and BCH codes}

Let $q$ be a power of a prime and let $n$ be a positive integer. 
An $[n,k, d]$ linear code $\C$ over $\gf(q)$ is a $k$-dimensional linear subspace of $\gf(q)$ 
with minimum Hamming distance $d$.  An $[n,k, d]$ linear code $\C$ over $\gf(q)$ is said to 
be distance-optimal if there is no $[n,k]$ linear code $\C$ over $\gf(q)$ with distance at least 
$d+1$.  Throughout this paper,  we use $\dim(\C)$ and $d(\C)$ to denote the dimension and 
minimum distance of the linear code $\C$, respectively. Let $A_i(\C)$ denote the number of 
codewords with Hamming weigh $i$ in $\C$. The polynomial $\sum_{j=0}^n A_i(\C)z^j$ is 
called the weight enumerator of $\C$.

Let $\textbf{b} = \mathbf (b_0, b_1, \ldots, b_{n-1})$ and $\textbf{c} = \mathbf (c_0, c_1, \ldots, c_{n-1})$ be two vectors in $\gf(q)^n$.  The standard inner product of the two vectors $\textbf{b}$ and $\textbf{c}$ is 
defined by 
$$ 
\textbf{b} \textbf{c}^T =\sum_{i=0}^{n-1} {b_i}{c_i}.
$$
The (Euclidean) dual code of $\mathcal C$, denoted by $\mathcal C^\perp$, is defined by
$$\mathcal C^{\perp} =\{\textbf{b} \in  \gf(q)^n: \textbf{b} \textbf{c}^T =\sum_{i=0}^{n-1} {b_i}{c_i} = 0 \ \  \forall \ \textbf{c} \in \mathcal C\}.  $$

An $[n,k, d]$ linear code $\C$ over $\gf(q)$ is said to be {\em cyclic} if
$(c_0,c_1, \ldots, c_{n-1}) \in \C$ implies $(c_{n-1}, c_0, c_1, \ldots, c_{n-2})
\in \C$.
By identifying a vector $(c_0,c_1, \ldots, c_{n-1}) \in \gf(q)^n$
with the polynomial 
$$
c_0+c_1x+c_2x^2+ \cdots + c_{n-1}x^{n-1} \in \gf(q)[x]/\langle x^n-1 \rangle,
$$
a code $\C$ of length $n$ over $\gf(q)$ corresponds to a subset of the quotient ring
$\gf(q)[x]/$ $\langle x^n-1 \rangle$.  
It is well known that 
a linear code $\C$ is cyclic if and only if the corresponding subset in $\gf(q)[x]/\langle x^n-1 \rangle$
is an ideal of the quotient ring $\gf(q)[x]/\langle x^n-1 \rangle$.
Notice that every ideal of $\gf(q)[x]/\langle x^n-1 \rangle$ is principal. Let $\C=\langle g(x) \rangle$ be a
cyclic code, where $g(x)$ is monic and has the smallest degree among all the
generators of $\C$. Then $g(x)$ must be unique and a divisor of $x^n-1$, and is called the {\em generator polynomial}
of $\C$
and $h(x)=(x^n-1)/g(x)$ is referred to as the {\em parity-check} polynomial of $\C$. 
The dual code $\C^\perp$ of the cyclic code $\C$ is also cyclic and has the generator polynomial 
$h^*(x):=h_k^{-1} x^kh(x^{-1})$, where $h_k$ is the coefficient of the term $x^k$ in $h(x)$.  

Recall that the generator polynomial of a cyclic code $\C$ of length $n$ must be a divisor of the 
polynomial $x^n-1$ over $\gf(q)$.  To study cyclic codes of length $n$ over $\gf(q)$, we must 
study the factorization of $x^n-1$ over $\gf(q)$.  To this end, we must define and deal with  
the $q$-cyclotomic cosets modulo $n$. 

Let $\gcd(n,q)=1$. 
Let $\Z_n$ denote  the set $\{0,1,2, \ldots, n-1\}$.  Let $s$ be an integer with $0 \leq s <n$. The \emph{$q$-cyclotomic coset of $s$ modulo $n$\index{$q$-cyclotomic coset modulo $n$}} is defined by
$$
C_s=\{s, sq, sq^2, \ldots, sq^{\ell_s-1}\} \bmod n \subseteq \Z_n,
$$
where $\ell_s$ is the smallest positive integer such that $s \equiv s q^{\ell_s} \pmod{n}$, and is the size of the
$q$-cyclotomic coset $C_s$. The smallest integer in $C_s$ is named as the \emph{coset leader\index{coset leader}} of $C_s$.
Let $\Gamma_{(n,q)}$ be the set of all the coset leaders. We have then $C_s \cap C_t = \emptyset$ for any two
distinct elements $s$ and $t$ in  $\Gamma_{(n,q)}$, and
\begin{eqnarray}\label{eqn-cosetPP}
	\bigcup_{s \in  \Gamma_{(n,q)} } C_s = \Z_n.
\end{eqnarray}
Consequently, the distinct $q$-cyclotomic cosets modulo $n$ partition $\Z_n$.

Let $m=\ord_{n}(q)$ be the order of $q$ modulo $n$, which is the smallest positive integer such that
$q^m \equiv 1 \pmod{n}$, and let $\alpha$ be a primitive element of $\gf(q^m)$. Put $\beta=\alpha^{(q^m-1)/n}$.
Then $\beta$ is a primitive $n$-th root of unity in $\gf(q^m)$. The minimal polynomial $\M_{\beta^s}(x)$
of $\beta^s$ over $\gf(q)$ is the monic polynomial of the smallest degree over $\gf(q)$ with $\beta^s$
as a root.  It is easily seen that this polynomial is given by
$$
\M_{\beta^s}(x)=\prod_{i \in C_s} (x-\beta^i) \in \gf(q)[x],
$$
which is irreducible over $\gf(q)$. It then follows from (\ref{eqn-cosetPP}) that
\begin{eqnarray}\label{eqn-canonicalfact}
	x^n-1=\prod_{s \in  \Gamma_{(n,q)}} \M_{\beta^s}(x),
\end{eqnarray}
which is the factorization of $x^n-1$ into irreducible factors over $\gf(q)$. This canonical factorization of $x^n-1$
over $\gf(q)$ is fundamental for the study of cyclic codes.

Let $\C$ be a cyclic code of length $n$  with generator polynomial $g(x)$.
The set $T=\{0 \le i \le n-1 : g(\beta^i)=0\}$ is called the
\emph{defining set} of $\mathcal C$ with respect to $\beta$.

Let $\delta$ be an integer with $2 \leq \delta \leq n$ and let $b$ be an integer.
A \emph{BCH code\index{BCH codes}} over $\gf(q)$
with length $n$ and \emph{designed distance} $\delta$, denoted by $\C_{(q,n,\delta,b)}$, is the cyclic code with
generator polynomial
\begin{eqnarray}\label{eqn-BCHgeneratorpolynomial}
	g_{(q,n,\delta,b)}(x)=\lcm(\M_{\beta^b}(x), \M_{\beta^{b+1}}(x), \ldots, \M_{\beta^{b+\delta-2}}(x)),
\end{eqnarray}
where the least common multiple is computed over $\gf(q)[x]$.

BCH codes over finite fields form a subclass of cyclic codes with many interesting properties and are closely related to many 
areas of mathematics.  In many cases BCH codes are the best linear codes.
For example, among all binary cyclic codes of odd length $n$ with $n \leq 125$ the best cyclic code is always a BCH code
except for two special cases \cite{Dingbook15}.  However,  certain BCH codes have bad parameters also.  
Reed-Solomon codes are also BCH codes and are widely used in communication
devices and consumer electronics.  This shows the importance of BCH codes in practice. 

Binary BCH codes were invented by Hocquenghem \cite{Hocq59} and independently by Bose and Ray-Chaudhuri \cite{BC60}.
They were generalized to BCH codes over finite fields by Gorenstein and Zierler in 1961 \cite{GZ61}.  
In the past sixty-five years, a lot of progress on the study of BCH codes has been made \cite{DingLiSurvey}.  
However, the parameters of most BCH codes are open.  

\subsection{The motivation of this paper}\label{sec-motivation}

The \emph{sphere packing bound}, also called the \emph{Hamming 
bound, } states that 
\begin{equation*}
\sum\limits_{i=0}^{\lfloor \frac{d-1}{2} \rfloor} \binom{n}{i} (q-1)^i \leq q^{n-k} 
\end{equation*} 
for any $[n, k, d]$ linear code over $\gf(q)$. 
This bound was published in 1950 by R. W. Hamming \cite{Hamming}. 

A code meeting the sphere packing bound is said to be \emph{perfect\index{perfect code}}.  
Since all perfect codes are known and very rare,  it is very interesting to look for distance-optimal linear codes 
with respect to the sphere packing bound. 
In the literature there are some infinite families of  distance-optimal linear codes with minimum distance 
at most 8 with respect to the sphere packing bound (see Section \ref{sec-known11}).   
However,  it is a 75-year-old open problem in coding theory whether there is an infinite family of distance-optimal linear codes over finite fields with arbitrarily large minimum distance with respect to the sphere packing bound.   
The motivation of this paper is this  75-year-old open problem in coding theory. 

\subsection{The objectives and methodology of this paper} 

The key objective of this paper is to settle the long-standing open problem in coding theory introduced 
in Section \ref{sec-motivation}.  Another objective is to construct infinite families of distance-optimal binary codes with small minimum distances with respect to the sphere packing bound.   

To achieve these objectives in this paper, 
we will study binary BCH codes with length $n$ of the following forms:
\begin{itemize}
\item $n=(2^{2s}+1)(2^s-1)$ for positive integer $s$.
\item $n=2^{2s}+2^s+1$ for positive integer $s$. 
\item $n=(2^s-1)/\lambda$ for positive integer $s$, where $\lambda$ is a constant divisor of $2^s-1$. 
\end{itemize} 
We will also study their related codes. 

\subsection{The achievements of this paper}

The main achievement of this paper is the settlement of the 75-year-old open problem in coding theory introduced 
in Section \ref{sec-motivation} by presenting infinitely many families of distance-optimal binary codes with arbitrarily large minimum distance. 
In addition,  several infinite families of distance-optimal binary codes with small minimum distances are presented.  
Two infinite families of binary five-weight codes are documented in this paper.  Some research problems are proposed.  

\subsection{The organization of this paper}

The rest of this paper is organized as follows.  Section \ref{sec-known} summarizes known infinite families of distance-optimal binary linear codes.  Section \ref{sec-mdistance} introduces some preliminary results.  
Section \ref{sec-type1} studies the BCH codes $\C_{(2,n,\delta,b)}$ and their related codes, where 
$n=(2^{2s}+1)(2^s-1)$ for positive integer $s$. 
Section \ref{sec-type2} treats the BCH codes $\C_{(2,n,\delta,b)}$ and their related codes, where 
$n=2^{2s}+2^s+1$ for positive integer $s$. 
Section \ref{sec-type3} investigates the BCH codes $\C_{(2,n,\delta,b)}$ and their related codes, where 
$n=(2^{2s}-1)/3$ for positive integer $s$.  
Section \ref{sec-newnew} presents infinitely many families of distance-optimal binary cyclic codes with variable dimension and arbitrarily large minimum distance.  
Section \ref{sec-open} proposes some open problems. 
Section \ref{sec-summary} summarizes this paper and makes some concluding remarks. 

\section{Known infinite families of distance-optimal binary linear codes}\label{sec-known} 

Any given $[n, k, d]$ code $\C$ over $\gf(q)$ can be extended into an  $[n+1, k, \overline{d}]$ code 
$\overline{\C}$ over $\gf(q)$, where 
\begin{eqnarray}\label{eqn-extendedcodegeneral}
\overline{\C}=\left\{(c_1, \ldots, c_n, c_{n+1}):  (c_1, \ldots, c_n) \in \C, \, c_{n+1}=-\sum_{i=1}^n c_i\right\}. 
\end{eqnarray} 
By definition, $\overline{d}=d$ or $\overline{d}=d+1$.  Even if $\C$ is not distance-optimal, the extended 
code $\overline{\C}$ could be distance-optimal.  This is a motivation of studying the extended linear codes. 

\subsection{Known infinite families of distance-optimal binary linear codes with respect to the sphere packing bound}\label{sec-known11}

To the best knowledge of the authors,  the known infinite families of distance-optimal binary linear codes 
with respect to the spher packing bound are the following: 
\begin{itemize} 
\item The binary codes $\overline{\C_{(2,2^m-1,2\ell+1,1)}}$ with parameters $[2^m, 2^m-1-\ell m, 2\ell+2]$ for $m \geq 6$ and $\ell \in \{1,2,3\}$ (see  \cite{Dingbook18} for a proof of the parameters of the codes. The distance optimality of the codes can be verified by the sphere packing bound).  

\item The binary codes $\C_{(2,2^m-1,2\ell+2,0)}$ with parameters $[2^m-1, 2^m-2-\ell m, 2\ell+2]$ for $m \geq 6$ and $\ell \in \{1,2,3\}$  (see  \cite{Dingbook18} for a proof of the parameters of the codes. The distance optimality of the codes can be verified by the sphere packing bound).

\item The binary codes $\C_{(2,(2^m-1)/\lambda,4,0)}$ with parameters $[(2^m-1)/\lambda, (2^m-1)/\lambda -1-m, 4]$ for $m$ being sufficiently large, where $\lambda < 2^{(m+2)/2}$ and $\lambda$ is a divisor of $2^m-1$ \cite{XCY24}.  

\item A family of binary codes with parameters $[2^{m+s} +2^s-2^m,  2^{m+s} +2^s-2^m-2m-2, 4]$ 
in \cite{HD20}. 

\item A family of binary codes with parameters $[2^{m}+2,  2^m-2m, 6]$ 
\cite{HD20}. 

\item Several families of binary codes with minimum distance $4$ and $6$ in 
\cite{WZD21}. 

\item A family of binary quasi-cyclic codes with parameters $[2(2^m-1)/\lambda, 2(2^m-1)/\lambda -2-m, 4]$ for $m$ being sufficiently large, where $\lambda < 2^{(m-2)/2}$ and $\lambda$ is a divisor of $2^m-1$  \cite{XCY24}.  

\item The binary codes $\C_{(2,(2^m-1)/3,8,0)}$ with parameters $[(2^m-1)/3, (2^m-1)/3 -1-3m, 8]$ for $m\geq 12$ being even  \cite{XCY24}.  

\item An infinite family of binary codes with parameters $[3(2^m-1), 3(2^m-1)-3m-1,6]$  for sufficiently large $m$  \cite{ChenWu}.  

\item The binary codes $\C_{(2,(2^m-1)/\lambda,6,0)}$ with parameters $[(2^m-1)/\lambda, (2^m-1)/\lambda -1-2m, 6]$ for $m$ being sufficiently large,  where $\lambda$ is a divisor of $2^m-1$ and $\lambda < 2^{(m-2)/2}$ \cite{ChenWu}.  
\end{itemize} 
Notice that the minimum distances of all the codes in the list above are restricted in the set $\{4,6,8\}$. 

\subsection{Known infinite families of distance-optimal binary linear codes with respect to the Griesmer bound} 

For $[n, k, d]$ linear codes over $\gf(q)$, the Griesmer bound says that 
\begin{eqnarray*}
n \geq \sum_{i=0}^{k-1} \left\lceil  \frac{d}{q^i} \right\rceil.  
\end{eqnarray*}
If the equality above holds, the code $\C$ is called a \emph{Griesmer code}.  By definition, a Griesmer code 
is distance-optimal with respect to the Griesmer bound.  

There are many constructions of binary Griesmer codes and many references about binary Griesmer codes. 
For detailed information on Griesmer codes, the reader is referred to \cite{Chen251, Hamada,Hell85,Hell81,HY16,SS65,Xieetal} and the references therein.

\section{Preliminaries}\label{sec-mdistance}

Let $\C$ be a cyclic code of length $n$ over $\gf(q)$ with generator polynomial
$
g(x)=\prod_{i \in T} (x-\beta^i),
$
where $\beta$ is an $n$-th root of unity over an extension field $\gf(q^m)$, $T$ is the union of some $q$-cyclotomic cosets modulo $n$, and is called
the \textit{defining set}\index{defining set} of $\C$ with respect to $\beta$.
We say that $T$ contains a set of $s$ consecutive integers if there is some
integer $a$ such that $\{a,a+1,\ldots,a+s-1\} \bmod n \subseteq T$.
The following
is a simple but very useful lower bound  (see \cite{BC60} and \cite{Hocq59}).

\begin{theorem}[BCH bound]\label{thm-BCHbound}
	Let $\C$ be a cyclic code of length $n$ over $\gf(q)$ with defining set $T$ and minimum distance $d$. Assume that $T$
	contains $\delta -1$ consecutive integers for some integer $\delta$. Then $d \geq \delta$.
\end{theorem}

The BCH bound depends on the choice of the primitive $n$-th root of unity $\beta$.
Different choices of the primitive root may yield different lower bounds. When applying
the BCH bound, it is important to choose the right primitive root so that the best lower bound is obtained.  However, it is open how
to choose such a primitive root.
By definition, for a BCH code $\C_{(q,n,\delta,b)}$ we have 
$$
d(\C_{(q,n,\delta,b)}) \geq \delta. 
$$ 

The following result is useful in some cases \cite{LiLiDing17cc}. 

\begin{theorem}\label{thm-LLD}
The BCH code $\C_{(q,n,\delta,b)}$ has minimum distance $d=\delta$ if $\delta$ divides $\gcd(n, b-1)$.
\end{theorem}

The following lemma was proved in \cite{AKS}.

\begin{lemma}\label{lem-AKS}
Let $n$ be a positive integer such that $q^{\lfloor m/2 \rfloor}<n \leq q^m-1$, where
$m=\ord_n(q)$. Then the $q$-cyclotomic coset $C_s=\{sq^j \bmod{n}: 0 \leq j \leq m-1\}$ has cardinality
$m$ for all $s$ in the range $1 \leq s \leq n q^{\lceil m/2 \rceil}/(q^m-1)$. In addition, every $s$ with
$s \not\equiv 0 \pmod{q}$ in this range is a coset leader.
\end{lemma}

\section{The first type of BCH binary codes and their related codes}\label{sec-type1} 

Let $s \geq 3$ be a positive integer and $n=(2^{2s}+1)(2^s-1)$.  Let $\alpha$ be a primitive element of $\gf(2^{4s})$ and $\beta=\alpha^{2^{s}+1}$. Then $\beta$ is an $n$-th primitive root of unity.  
Let $\m_{\beta}(x)$ denote the minimal polynomial of $\beta$ over $\gf(2)$.  Throughout this section, we fix the notation.  We will need the following lemma in this section. 

\begin{lemma}\label{lem-sec3} 
Recall that $C_i$ denotes the $2$-cyclotomic coset modulo $n$ containing $i$.  Let $s \geq 3$. Then the following hold. 
\begin{itemize} 
\item $\ord_n(2)=4s$. 
\item All the odd $i$ with $1 \leq i \leq 2^s-1$  are all coset leaders of the corresponding $2$-cyclotomic cosets $C_i$.  
\item $|C_i|=4s$ for all  the odd $i$ with $1 \leq i \leq 2^s-1$.  
\end{itemize}
\end{lemma} 

\begin{proof} 
The conclusion of the first part is known (see \cite{XCY24} for example).  
The conclusion of the remaining parts then follow from Lemma \ref{lem-AKS}.  
\end{proof}

\subsection{The parameters of the codes  $\C_{(2,n,3,1)}$ and their duals}

In this subsection, we study the codes  $\C_{(2,n,3,1)}$ and their duals and do some preparations for investigating the codes $\overline{\C_{(2,n,3, 1)}}$ in Section \ref{sec-222}. 

\begin{theorem}\label{thm-one1}
Let $s \geq 2$.  Then code $\C_{(2,n,  3,   1)}$  has parameters $[n, n-4s, 3]$ and $\C_{(2,n,  3,   1)}^\perp$ has parameters 
$[n, 4s, 2^{3s-1}-2^{2s-1}]$ and weight enumerator 
$$
1+(2^{4s}-1-n)z^{ 2^{3s-1}-2^{2s-1}}  +nz^{2^{3s-1}}. 
$$
\end{theorem}

\begin{proof} 
By definition and Lemma \ref{lem-sec3},  $\C_{(2,n,  3,   1)}$ is the binary cyclic code of length $n$ with generator polynomial 
$\m_{\beta}(x)$. Then by Lemma \ref{lem-sec3} we have 
$$
\dim( \C_{(2,n,  3,   1)})=n-|C_1|=n-4s
$$ 
and 
$$
\dim( \C_{(2,n,  3,   1)}^\perp )=4s. 
$$
By definition, $\C_{(2,n,  3,   1)}^\perp$ is an irreducible cyclic code.

Let $\gamma=\beta^{-1}$.  Then $\gamma$ is an $n$-th primitive root of unity.   
Then the trace representation of $\C_{(2,n,  3,   1)}^\perp$ is given by 
\begin{eqnarray}
\C_{(2,n,  3,   1)}^\perp=\{\bc_a=(\tr(a\gamma^{i}))_{i=0}^{n-1}: a \in \gf(2^{4s})\}, 
\end{eqnarray}
where $\tr$ is the absolute trace function.  Then $\C_{(2,n,  3,   1)}^\perp$ is an irreducible cyclic code 
and its weight enumerator follows from Theorem 23 in \cite{DingYangSurv}.

It follows from the BCH bound that the minimum distance of $\C_{(2,n,  3,   1)}$ is at least $3$.  
By the fourth Pless power moment \cite{HuffmanPless03bk}, we have 
\begin{eqnarray*}
8\sum_{j=0}^n j^3 A_j( \C_{(2,n,  3,   1)}^\perp) =2^{4s}(n^2(n+3)-6A_3(\C_{(2,n,  3,   1)})). 
\end{eqnarray*} 
Solving the equation above yields that 
$$ 
A_3(\C_{(2,n,  3,   1)}) =\frac{(u-2)(u-1)(u^2+1)}{6}>0, 
$$
where $u=2^s$.  As a result,  $d(\C_{(2,n,  3,   1)})=3.$ 
\end{proof}

\begin{example}{\rm
Let $s=2$. Then the code $\C_{(2,n,  3,   1)}$ has parameters $[51,43,3]$ and has the same parameters 
as the optimal binary cyclic code with length and dimension $[51,43]$ \cite[Appendix A]{Dingbook15}. 
 The binary code $\C_{(2,n,  3,   1)}^\perp$ has parameters 
$[51,8,24]$ and weight enumerator 
$$
1+204z^{24} +51z^{32}. 
$$ 
In addition,  $\C_{(2,n,  3,   1)}^\perp$ has the same parameters 
as the optimal binary cyclic code with length $51$ and dimension $8$ \cite[Appendix A]{Dingbook15}.   }
\end{example}

\begin{example}{\rm
Let $s=3$. Then the code $\C_{(2,n,  3,   1)}$ has parameters $[455,443,3]$.  Furthermore,  $\C_{(2,n,  3,   1)}^\perp$ has parameters 
$[455,12,224]$ and weight enumerator 
$$
1+3640z^{224}+455z^{256}.
$$}
\end{example}

\subsection{The extended codes $ \overline{\C_{(2,n,  3,   1)}}$ of the codes  $\C_{(2,n,  3,   1)}$}\label{sec-222}

Our first infinite family of distance-optimal binary codes is given in the following theorem.

\begin{theorem}\label{thm-one2}
Let $s \geq 2$.  Then the extended code $\overline{\C_{(2,n,  3,   1)}}$  has parameters $[n+1, n-4s, 4]$
 and is distance-optimal with respect to the sphere packing bound. 
 Furthermore, 
 $\overline{\C_{(2,n,  3,   1)}}^\perp$ has parameters 
$[n+1, 4s+1,  n+1-2^{3s-1}]$ and weight enumerator 
\begin{eqnarray*}
1+nz^{n+1-2^{3s-1}} + (2^{4s}-1-n)z^{ 2^{3s-1}-2^{2s-1}}  +  \\
(2^{4s}-1-n)z^{ n+1-2^{3s-1}+2^{2s-1}}+nz^{2^{3s-1}} + z^{n+1}.
\end{eqnarray*} 
\end{theorem}

\begin{proof} 
Since $3$ is the minimum distance of $\C_{(2,n,  3,   1)}$ and is odd, 
the parameters of $\overline{\C_{(2,n,  3,   1)}}$ follow from the parameters of 
$\C_{(2,n,  3,   1)}$ directly.  

We now prove the distance-optimality of the code $\overline{\C_{(2,n,  3,   1)}}$.  
Note that $s \geq 2$. We have 
\begin{eqnarray*}
\sum_{i=0}^2 \binom{n+1}{i} = \frac{n^2+3n+4}{2}  > 2^{4s+1}.  
\end{eqnarray*} 
By the sphere packing bound, there is no binary linear code with parameters $[n+1, n-4s, d\geq 5]$.

By Theorem \ref{thm-one1},  $\C_{(2,n,  3,   1)}^\perp$ has only even weights and
the following hold:
\begin{itemize}
\item The extended coordinate of each codeword in the extended code  $\overline{\C_{(2,n,  3,   1)}^\perp}$ is always 0. 
\item The code $\overline{\C_{(2,n,  3,   1)}^\perp}$ has parameters  
$[n+1, 4s, 2^{3s-1}-2^{2s-1}]$ and weight enumerator 
$$
1+(2^{4s}-1-n)z^{ 2^{3s-1}-2^{2s-1}}  +nz^{2^{3s-1}}. 
$$
\end{itemize} 
It is then easily seen that 
\begin{eqnarray*}
\overline{\C_{(2,n,  3,   1)}}^\perp= \overline{\C_{(2,n,  3,   1)}^\perp} \cup (\bone +  \overline{\C_{(2,n,  3,   1)}^\perp} ).  
\end{eqnarray*}
Then the desired weight enumerator and the parameters of the code  $\overline{\C_{(2,n,  3,   1)}}^\perp$  follow. 
\end{proof}

\subsection{The distance-optimal codes $\C_{(2,n,  6,   0)}$}

In this subsection, our second infinite family of distance-optimal binary codes is given in the following theorem.
\begin{theorem}\label{thm-one2-add}
    Let $s \geq 8$.  Then $\C_{(2,n, 6, 0)}$  has parameters $[n, n-8s-1, 6]$ and is distance-optimal with respect to the sphere packing bound.   
\end{theorem}
\begin{proof}
    By definition and Lemma \ref{lem-sec3},  $\C_{(2,n, 6, 0)}$ is the binary cyclic code of length $n$ with generator polynomial $\m_{\beta^0}(x)\m_{\beta}(x)\m_{\beta^3}(x)$. Thus, 
    $$\dim( \C_{(2,n,6, 0)})=n-|C_0|-|C_1|-|C_3|=n-8s-1.$$ 
    It follows from the BCH bound that $d(\C_{(2,n, 6, 0)}) \geq 6$.  When $s\geq 8$, it can be verified that 
    \begin{eqnarray*}
    \lefteqn{\sum_{i=0}^3 \binom{n}{i} -2^{8s+1}} \\
    &=& 1+n+n(n-1)/2+n(n-1)(n-2)/6 -2^{8s+1}\\
    &\geq & (2^{3s-1}\times2^{3s-1}\times2^{3s-1})/6 -2^{8s+1} \\
    &> & 2^{9s-6} -2^{8s+1}\\
    &>&0, 
    \end{eqnarray*}
    According to the sphere packing bound, we have $d(\C_{(2,n, 6, 0)}) \leq 6$.  Consequently, $d(\C_{(2,n, 6, 0)}) =6$ and code $\C_{(2,n, 6, 0)}$ is distance-optimal when $s\geq 8$.  
\end{proof}

\subsection{The codes $\C_{(2,n,  5,   1)}$  and their extended codes} 

In this subsection, we study the parameters of the code $\C_{(2,n,  5,   1)}$ and its extended code 
$ \overline{\C_{(2,n,  5,   1)}}$.  
Our third infinite family of distance-optimal binary codes is documented in the following theorem. 

\begin{theorem}\label{thm-s31}
Let $s \geq 4$.  Then $\C_{(2,n,  5,   1)}$  has parameters $[n, n-8s, 5 \leq d \leq 6]$.  The extended code 
$ \overline{\C_{(2,n,  5,   1)}}$ has parameters 
$[n+1, n-8s, 6]$ and is distance-optimal with respect to the sphere packing bound for $s\geq4$.   
\end{theorem}

\begin{proof}
By definition and Lemma \ref{lem-sec3},  $\C_{(2,n,  5,   1)}$ is the binary cyclic code of length $n$ with generator polynomial 
$\m_{\beta}(x) \m_{\beta^3}(x) $, where $\gcd(\m_{\beta}(x), \m_{\beta^3}(x))=1$. Then by Lemma \ref{lem-sec3} we have 
$$
\dim( \C_{(2,n,  5,   1)})=n-|C_1|-|C_3|=n-8s
$$ 
and 
$$
\dim( \C_{(2,n,  5,   1)}^\perp )=8s. 
$$
Consequently,  $ \dim(\overline{\C_{(2,n,  5,   1)}})=n-8s$.

It follows from the BCH bound that $d(\C_{(2,n,  5,   1)}) \geq 5$.  
When $s\geq 4$, it can be verified that 
\begin{eqnarray*}
\lefteqn{\sum_{i=0}^3 \binom{n}{i} -2^{8s}} \\
&=& u(u-1)(u^2+1)(u^5 - 8u^4 - 3u^3 - 4u^2 + 3u - 8)/6 \\
&\geq & (16u^4 - 8u^4 - 3u^3 - 4u^2 + 3u - 8)/6 \\
&=& (8u^4 - 3u^3 - 4u^2 + 3u - 8)/6 \\
&\geq & (128u^3 - 3u^3 - 4u^2 + 3u - 8)/6 \\
&= & (125u^3 - 4u^2 + 3u - 8)/6 \\
&>&0, 
\end{eqnarray*}
where $u=2^s \geq 16$.  
It then follows from the sphere packing bound that $d(\C_{(2,n,  5,   1)}) \leq 6$.  Consequently, 
$5 \leq d(\C_{(2,n,  5,   1)}) \leq 6$.  

In both cases that $d(\C_{(2,n,  5,   1)})=5$ and $d(\C_{(2,n,  5,   1)})=6$, we have $d( \overline{\C_{(2,n,  5,   1)}})=6$. Hence, the parameters of the code $ \overline{\C_{(2,n,  5,   1)}}$  follow. 
When $s \geq 4$, it can be verified that 
\begin{eqnarray*}
\lefteqn{ \sum_{i=0}^3 \binom{n+1}{i} -2^{8s+1}} \\
&=& (u-1)(u^2+1)(u^6 - 14u^5 - 9u^4 - u^3 - 11u - 6)/6 \\
&\geq& (u^6 - 14u^5 - 9u^4 - u^3 - 11u - 6)/6 \\
&\geq& (2u^5 - 9u^4 - u^3 - 11u - 6)/6 \\
&\geq& (23u^4 - u^3 - 11u - 6)/6 \\
&\geq& (735u^3 - 11u - 6)/6 \\
&>&0, 
\end{eqnarray*}
where $u=2^s \geq 16$.  
Hence,  the code $ \overline{\C_{(2,n,  5,   1)}}$ is distance-optimal with respect to the sphere packing bound 
for $s \geq 4$. 
\end{proof} 

\begin{conj}
The minimum distance of $\C_{(2,n,  5,   1)}$ is $5$ for $s \geq 2$. 
\end{conj} 

\begin{remark}
Notice that $n=(2^{2s}+1)(2^s-1)$.  We have $n \equiv 0 \pmod{5}$ when $s$ is odd or $s \equiv 0 \pmod{4}$. 
It then follows from Theorem \ref{thm-LLD} that $d(\C_{(2,n,  5,   1)}) = 5$ if $s$ is odd or $s \equiv 0 \pmod{4}$.  It would be good if this conjecture can be confirmed in the remaining case  $s \equiv 2 \pmod{4}$.  
\end{remark} 

\begin{example} {\rm
Let $s = 2$.  Then $\C_{(2,n,  5,   1)}$  has parameters $[51,35,5]$,  which are the parameters 
of the optimal binary cyclic code with length $51$ and dimension $35$ 
\cite[Appendix A]{Dingbook18}.  Its extended code 
$ \overline{\C_{(2,n,  5,   1)}}$ has parameters $[52,35,6]$, while the best known binary linear code 
has parameters $[52,35,7]$ \cite{G}. }
\end{example}

\begin{example} {\rm
Let $s=3$. Then $\C_{(2,n,  5,   1)}$  has parameters $[455,431,5]$ and its extended code 
$ \overline{\C_{(2,n,  5,   1)}}$ has parameters $[456,431,6]$. }
\end{example}

\section{The second type of BCH cyclic codes and their related codes}\label{sec-type2}  

Let $s\geq 2$ be an integer. Let $n=2^{2s}+2^s+1$.  Let $\alpha$ be a primitive element 
of $\gf(2^{3s})$ and $\beta=\alpha^{2^s-1}$.  Then $\beta$ is an $n$-th primitive root of unity in $\gf(2^{3s})$.  Let $\m_{\beta^i}(x)$ denote the minimal polynomial of $\beta^i$ over $\gf(2)$. We fix the notation throughout this section.  We will need the following lemma in this section. 

\begin{lemma}\label{lem-sec4} 
Recall that $C_i$ denotes the $2$-cyclotomic coset modulo $n$ containing $i$.  Let $s \geq 3$. Then the following hold. 
\begin{itemize} 
\item $\ord_n(2)=3s$. 
\item All the elements $i$ in $\{1,3\}$ are the coset leaders of the $2$-cyclotomic cosets $C_i$.  
\item $|C_i|=3s$ for all $i \in \{1,3\}$.  
\end{itemize}
\end{lemma} 

\begin{proof} 
The conclusion of the first part is known (see \cite{XCY24} for example).  
The conclusions of the remaining parts then follow from Lemma \ref{lem-AKS}.  
\end{proof}

Our fourth infinite family of distance-optimal binary codes is described in the following theorem. 

\begin{theorem}\label{thm-two1}
The code $\C_{(2,n,3,1)}$ has parameters $[n, n-3s,3\leq d \leq 4]$. 
The extended code $\overline{\C_{(2,n,3,1)}}$ has parameters $[n+1, n-3s,4]$ and is distance-optimal with respect to the sphere packing bound for $s \geq 2$. 
\end{theorem} 

\begin{proof} 
It follows from the BCH bound that the minimum distance of $\C_{(2,n,  3,   1)}$ is at least $3$.  It can be easily verified that 
\begin{eqnarray*}
\sum_{i=0}^2 \binom{n}{i} -2^{3s} =(u^2 - u + 4)(u^2 + u + 1)/2>0, 
\end{eqnarray*} 
where $u=2^s$. 
It then follows from the sphere packing bound that $d(\C_{(2,n,  3,   1)})\leq 4.$ 
Hence, we have $3 \leq d(\C_{(2,n,  3,   1)})\leq 4.$

In both cases that $d(\C_{(2,n,  3,   1)}) =3$ and $d(\C_{(2,n,  3,   1)}) = 4$, 
we have $\overline{\C_{(2,n,  3,   1)}}=4$.  Then the desired parameters for $\overline{\C_{(2,n,  3,   1)}}$ follow.  It can be verified that 
\begin{eqnarray*}
\sum_{i=0}^2 \binom{n+1}{i} -2^{3s+1}  = (u^2 - 3u + 8) (u^2 + u + 1)/2 >0, 
\end{eqnarray*} 
where $u=2^s \geq 4$.  
It then follows from the sphere packing bound that $\overline{\C_{(2,n,  3,   1)}}$ is distance-optimal. 
\end{proof}

\begin{conj}
The minimum distance of $\C_{(2,n,  3,   1)}$ is $3$ for $s \geq 2$. 
\end{conj} 

\begin{remark}
Notice that $n=2^{2s}+2^s+1$.  We have $n \equiv 0 \pmod{3}$ when $s$ is even.  
It then follows from Theorem \ref{thm-LLD} that $d(\C_{(2,n,  3,   1)}) = 3$ if $s$ is even. 
It would be good if this conjecture can be confirmed in the odd $s$ case.  
\end{remark} 

\begin{example} {\rm
Let $s = 2$.  Then $\C_{(2,n,  3,   1)}$  has parameters $[21,15,3]$,  while the optimal binary cyclic code 
has parameters $[21,15,4]$ \cite[Appendix A]{Dingbook18}.  Its extended code 
$ \overline{\C_{(2,n,  3,   1)}}$ has parameters $[22,15,4]$. }
\end{example}

\begin{example} {\rm
Let $s=3$.  Then $\C_{(2,n,  3,   1)}$  has parameters $[73,64,3]$ and is a optimal binary cyclic code 
\cite[Appendix A]{Dingbook18}. 
The extended code 
$ \overline{\C_{(2,n,  3,   1)}}$ has parameters $[74,64,4]$. }
\end{example}

\section{The third type of binary BCH codes and their related codes}\label{sec-type3}  

Let $n=(4^s-1)/3$ and $s \geq 4$. Let $\alpha$ be a primitive element of $\gf(2^{2s})$ and 
$\beta=\alpha^3$.  Let $\m_{\beta^i}(x)$ denote the 
minimal polynomial of $\beta^i$ over $\gf(2)$. We fix the notation throughout this section.   
We will need the following lemma in this section. 

\begin{lemma}\label{lem-sec5} 
Recall that $C_i$ denotes the $2$-cyclotomic coset modulo $n$ containing $i$.  Let $s \geq 4$. Then the following hold. 
\begin{itemize} 
\item $\ord_n(2)=2s$. 
\item All the odd $i$ with $1 \leq i \leq 2^s/3$ are all coset leaders of the $2$-cyclotomic cosets $C_i$.  
\item $|C_i|=2s$ for all odd $i$ with $1 \leq i \leq 2^s/3$.  
\end{itemize}
\end{lemma} 

\begin{proof} 
The conclusion of the first part is known (see \cite{XCY24} for example).  
The conclusions of the remaining parts then follow from Lemma \ref{lem-AKS}.  
\end{proof}

\subsection{The parameters of the codes  $\C_{(2,n,  3,   1)}$ and their duals}

In this subsection, we study the codes  $\C_{(2,n,  3,   1)}$ and their duals and have the following results.  
This is a preparation of treating the codes $ \overline{\C_{(2,n,  3,   1)}}$ in Section \ref{sec-223}. 

\begin{theorem}\label{thm-five1}
Let $s \geq 4$.  Then code $\C_{(2,n,  3,   1)}$  has parameters $[n, n-2s, 3]$ and $\C_{(2,n,  3,   1)}^\perp$ has parameters 
$[n, 2s, d]$ and weight enumerator 
$$
1+nz^{\frac{2^{2s-1}+(-1)^s 2^s}{3}}  +2nz^{\frac{2^{2s-1}-(-1)^s 2^{s-1}}{3}}, 
$$
where 
$$
d=\min\left\{ \frac{2^{2s-1}+(-1)^s 2^s}{3},  \frac{2^{2s-1}-(-1)^s 2^{s-1}}{3}\right\} 
$$
\end{theorem}

\begin{proof} 
By definition and Lemma \ref{lem-sec5},  $\C_{(2,n,  3,   1)}$ is the binary cyclic code of length $n$ with generator polynomial 
$\m_{\beta}(x)$. Then by Lemma \ref{lem-sec5} we have 
$$
\dim( \C_{(2,n,  3,   1)})=n-|C_1|=n-2s
$$ 
and 
$$
\dim( \C_{(2,n,  3,   1)}^\perp )=2s. 
$$
By definition, $\C_{(2,n,  3,   1)}^\perp$ is an irreducible cyclic code.

Let $\gamma=\beta^{-1}$.  Then $\gamma$ is an $n$-th primitive root of unity.   
Then the trace representation of $\C_{(2,n,  3,   1)}^\perp$ is given by 
\begin{eqnarray}
\C_{(2,n,  3,   1)}^\perp=\{\bc_a=(\tr(a\gamma^{i}))_{i=0}^{n-1}: a \in \gf(2^{2s})\}, 
\end{eqnarray}
where $\tr$ is the absolute trace function.  Then $\C_{(2,n,  3,   1)}^\perp$ is an irreducible cyclic code 
and its weight enumerator follows from Theorem 23 in \cite{DingYangSurv}.  

Note that $3n+1=2^{2s}$. 
It can be verified that 
\begin{eqnarray*}
\lefteqn{8\sum_{j=0}^n j^3 A_j(\C_{(2,n,  3,   1)}^\perp)} \\ 
&=& 8 \left[ n \left(\frac{2^{2s-1}+(-1)^s 2^s}{3} \right)^3 + 2n \left(\frac{2^{2s-1}-(-1)^s 2^{s-1}}{3} \right)^3    \right] \\
&<& (3n+1)n^2(n+3). 
\end{eqnarray*} 
In fact,  when $s$ is odd, we have 
\begin{eqnarray*}
&& (3n+1)n^2(n+3) - 8 \left[ n \left(\frac{2^{2s-1}+(-1)^s 2^s}{3} \right)^3 + 2n \left(\frac{2^{2s-1}-(-1)^s 2^{s-1}}{3} \right)^3    \right]  \\
&=& (u^6 + 2u^5 - 9u^4 - 2u^3 + 8u^2 )/27 \\
&=& (u^2(u-2) (u-1) (u+1) (u+4) )/27\\
&>&0, 
\end{eqnarray*}
where $u=3^s\geq 27$. 
When $s$ is even, we have 
\begin{eqnarray*}
&& (3n+1)n^2(n+3) - 8 \left[ n \left(\frac{2^{2s-1}+(-1)^s 2^s}{3} \right)^3 + 2n \left(\frac{2^{2s-1}-(-1)^s 2^{s-1}}{3} \right)^3    \right]  \\
&=& (u^6 - 2u^5 - 9u^4 + 2u^3 + 8u^2) /27\\
&=& (u^2(u-4)(u-1)(u+1)(u+2))/27\\
&>&0, 
\end{eqnarray*}
where $u=3^s\geq 27.$. 

By the BCH bound, we have $d(\C_{(2,n,  3,   1)}) \geq 3.$ 
By the fourth Pless power moment,  \cite{HuffmanPless03bk},  we have then 
\begin{eqnarray*}
\lefteqn{8\sum_{j=0}^n j^3 A_j(\C_{(2,n,  3,   1)}^\perp)} \\ 
&=& 2^{2s} (n^2(n+3)-6A_3(\C_{(2,n,  3,   1)})) \\ 
&=& (3n+1) (n^2(n+3)-6A_3(\C_{(2,n,  3,   1)})). 
\end{eqnarray*} 
Thus, $A_3(\C_{(2,n,  3,   1)} )>0$.  Summarizing the discussions above, we conclude that 
$d(\C_{(2,n,  3,   1)} )=3$. 
%It then follows from the sphere packing bound that $\C_{(2,n,  3,   1)}$ is almost distance-optimal. 
\end{proof}

\begin{example}{\rm
Let $s=3$. Then the code $\C_{(2,n,  3,   1)}$ has parameters $[21,15,3]$ and has the same parameters 
as the optimal binary cyclic code with length and dimension $[21,15]$ \cite[Appendix A]{Dingbook15}. 
 The binary code $\C_{(2,n,  3,   1)}^\perp$ has parameters 
$[21,6,8]$ and weight enumerator 
$$
1+21z^{8} +42z^{12}. 
$$ 
In addition,  $\C_{(2,n,  3,   1)}^\perp$ has the same parameters 
as the optimal binary cyclic code with length $21$ and dimension $6$ \cite[Appendix A]{Dingbook15}.   }
\end{example}

\begin{example}{\rm
Let $s=4$. Then the code $\C_{(2,n,  3,   1)}$ has parameters $[85,77,3]$, which are the parameters of 
the optimal binary cyclic code with length and dimension $[85,77]$ \cite[Appendix A]{Dingbook15}.   Furthermore,  $\C_{(2,n,  3,   1)}^\perp$ has parameters 
$[85,8,40]$ and weight enumerator 
$$
1+170z^{40}+85z^{48}.
$$}
\end{example}

\subsection{The extended codes $ \overline{\C_{(2,n,  3,   1)}}$ of the codes  $\C_{(2,n,  3,   1)}$}\label{sec-223}

Our fifth infinite family of distance-optimal binary codes is given in the following theorem. 

\begin{theorem}\label{thm-five2}
Let $s \geq 4$.  Then the extended code $\overline{\C_{(2,n,  3,   1)}}$  has parameters $[n+1, n-2s, 4]$
 and is distance-optimal with respect to the sphere packing bound. 
 Furthermore, 
 $\overline{\C_{(2,n,  3,   1)}}^\perp$ has parameters 
$[n+1, 2s+1]$ and weight enumerator 
\begin{eqnarray*}
1+nz^{\frac{2^{2s-1}+(-1)^s 2^s}{3}}  +2nz^{\frac{2^{2s-1}-(-1)^s 2^{s-1}}{3}} \\
+nz^{n+1-\frac{2^{2s-1}+(-1)^s 2^s}{3}}  +2nz^{n+1-\frac{2^{2s-1}-(-1)^s 2^{s-1}}{3}} 
+ z^{n+1}.
\end{eqnarray*} 
\end{theorem}

\begin{proof} 
Since $3$ is the minimum distance of $\C_{(2,n,  3,   1)}$ and is odd, 
the parameters of $\overline{\C_{(2,n,  3,   1)}}$ follow from the parameters of 
$\C_{(2,n,  3,   1)}$ directly.  

We now prove the distance-optimality of the code $\overline{\C_{(2,n,  3,   1)}}$.  
Note that $s \geq 2$. We have 
\begin{eqnarray*}
\sum_{i=0}^2 \binom{n+1}{i} = \frac{n^2+3n+4}{2}  > 2^{2s+1}=6n+2.  
\end{eqnarray*} 
By the sphere packing bound, there is no binary linear code with parameters $[n+1, n-2s, d\geq 5]$.

By Theorem \ref{thm-five1},  $\C_{(2,n,  3,   1)}^\perp$ has only even weights and
the following hold:
\begin{itemize}
\item The extended coordinate of each codeword in the extended code  $\overline{\C_{(2,n,  3,   1)}^\perp}$ is always 0. 
\item The code $\overline{\C_{(2,n,  3,   1)}^\perp}$ has parameters  
$[n+1, 2s]$ and weight enumerator 
$$
1+nz^{\frac{2^{2s-1}+(-1)^s 2^s}{3}}  +2nz^{\frac{2^{2s-1}-(-1)^s 2^{s-1}}{3}}. 
$$ 
\end{itemize} 
It is then easily seen that 
\begin{eqnarray*}
\overline{\C_{(2,n,  3,   1)}}^\perp= \overline{\C_{(2,n,  3,   1)}^\perp} \cup (\bone +  \overline{\C_{(2,n,  3,   1)}^\perp} ).  
\end{eqnarray*}
The remaining desired conclusions follow. 
\end{proof} 

\begin{example} {\rm
Let $s=3$. Then the extended code $\overline{\C_{(2,n,  3,   1)}}$  has parameters $[22,15,4]$ and the code 
 $\overline{\C_{(2,n,  3,   1)}}^\perp$ has parameters 
$[22,7,8]$ and weight enumerator  
$$ 
1 + 21z^{8}  + 42z^{10} + 42z^{12} + 21z^{14} + z^{22}. 
$$}
\end{example} 

\begin{example} {\rm
Let $s=4$. Then the extended code $\overline{\C_{(2,n,  3,   1)}}$  has parameters $[86,77,4]$ and the code 
 $\overline{\C_{(2,n,  3,   1)}}^\perp$ has parameters 
$[86,9,38]$ and weight enumerator  
$$ 
1 + 85z^{38}  + 170z^{40} +170 z^{46} +85 z^{48} + z^{86}. 
$$}
\end{example}

\subsection{The codes $\C_{(2,n,  5,   1)}$  and their extended codes} 

In this section, we study the parameters of the code $\C_{(2,n,  5,   1)}$ and its extended code 
$ \overline{\C_{(2,n,  5,   1)}}$.  Our sixth infinite family of distance-optimal binary codes is introduced in the following theorem.

\begin{theorem}\label{thm-s61}
Let $s \geq 4$.  Then $\C_{(2,n,  5,   1)}$  has parameters $[n, n-4s, 5\leq d \leq 6]$.  The extended code 
$ \overline{\C_{(2,n,  5,   1)}}$ has parameters 
$[n+1, n-4s, 6]$ and is distance-optimal with respect to the sphere packing bound for $s\geq 5$.   
\end{theorem}

\begin{proof}
By definition and Lemma \ref{lem-sec5},  $\C_{(2,n,  5,   1)}$ is the binary cyclic code of length $n$ with generator polynomial 
$\m_{\beta}(x) \m_{\beta^3}(x) $, where $\gcd(\m_{\beta}(x), \m_{\beta^3}(x))=1$. Then by Lemma \ref{lem-sec5} we have 
$$
\dim( \C_{(2,n,  5,   1)})=n-|C_1|-|C_3|=n-4s
$$ 
and 
$$
\dim( \C_{(2,n,  5,   1)}^\perp )=4s. 
$$
Consequently,  $ \dim(\overline{\C_{(2,n,  5,   1)}})=n-4s$.

It follows from the BCH bound that $d(\C_{(2,n,  5,   1)}) \geq 5$.  
When $s\geq 4$,  we have $n \geq 341$. It can verified that 
\begin{eqnarray*}
\lefteqn{  \sum_{i=0}^3 \binom{n}{i} -2^{4s} } \\  
&=& \sum_{i=0}^3 \binom{n}{i} -(3n+1)^2 \\ 
&=& \frac{n(n^2 - 54n - 31)}{6} \\
&\geq & \frac{n( 287n - 64)}{6} \\
&>& 0. 
\end{eqnarray*} 
It then follows from the sphere packing bound that $d(\C_{(2,n,  5,   1)}) \leq 6$.   
Consequently, we have that  $5 \leq d(\C_{(2,n,  5,   1)}) \leq 6$. 

In both cases that $d(\C_{(2,n,  5,   1)})=5$ and $d(\C_{(2,n,  5,   1)})=6$, we have $d( \overline{\C_{(2,n,  5,   1)}})=6$. Hence, the parameters of the code $ \overline{\C_{(2,n,  5,   1)}}$  follow.    
When $s \geq 5$,  we have $n \geq 341$. It can be verified that 
\begin{eqnarray*}
\lefteqn{ \sum_{i=0}^3 \binom{n+1}{i} -2^{4s+1} }\\ 
&=& \sum_{i=0}^3 \binom{n+1}{i} -2(3n+1)^2 \\ 
&=& \frac{n(n^2 - 105n - 64)}{6} \\
&\geq & \frac{n( 236n - 64)}{6} \\
&>& 0. 
\end{eqnarray*} 
Hence,  the code $ \overline{\C_{(2,n,  5,   1)}}$ is distance-optimal with respect to the sphere packing bound 
for $s \geq 5$. 
\end{proof}

\begin{conj}
The minimum distance of $\C_{(2,n,  5,   1)}$ is $5$ for $s \geq 2$. 
\end{conj} 

\begin{remark}
Notice that $n=(4^s-1)/3$.  We have $n \equiv 0 \pmod{5}$ when $s$ is even. 
It then follows from Theorem \ref{thm-LLD} that $d(\C_{(2,n,  5,   1)}) = 5$ if $s$ is even.  It would be good if this conjecture can be confirmed in the odd $s$ case.  
\end{remark} 

\begin{example} {\rm
Let $s = 3$.  Then $\C_{(2,n,  5,   1)}$  has parameters $[21,12,5]$ and is optimal.  Its extended code 
$ \overline{\C_{(2,n,  5,   1)}}$ has parameters $[22,12,6]$ and is optimal. }
\end{example}

\begin{example} {\rm
Let $s=4$. Then $\C_{(2,n,  5,   1)}$  has parameters $[85,69,5]$, 
while the optimal binary linear code has parameters  $[85,69,6]$. The extended code 
$ \overline{\C_{(2,n,  5,   1)}}$ has parameters $[86,69,6]$, the best known parameters \cite{G}. }
\end{example}

\section{Infinitely many families of distance-optimal binary cyclic codes with arbitrarily large minimum distance}\label{sec-newnew} 

Let $s \geq 3$ be a positive integer and $n=(2^s-1)/\lambda$, where $\lambda < 2^{\lfloor s/2 \rfloor}$ and $\lambda$ is a constant divisor of $2^s-1$.  Hence, $\lambda$ does not vary according to $s$.  Let $\alpha$ be a primitive element of $\gf(2^s)$, and let $\beta=\alpha^{\lambda}$. 
Then $\beta$ is an $n$-th primitive root of unity.  
Let $\m_{\beta}(x)$ denote the minimal polynomial of $\beta$ over $\gf(2)$.  Throughout this section, we fix the notation.  We will need the following lemma in this section. 

\begin{lemma}\label{lem-secfinal} 
Recall that $C_i$ denotes the $2$-cyclotomic coset modulo $n$ containing $i$.  
Then the following hold. 
\begin{itemize} 
\item $\ord_n(2)=s$. 
\item For each odd $i$ with $1 \leq i \leq 2^{\lceil s/2\rceil}/\lambda$, we have $|C_i|=s$.  
\item All the odd $i$ with  $1 \leq i \leq 2^{\lceil s/2\rceil}/\lambda$ are the coset leaders of these $C_i$.  
\end{itemize}
\end{lemma} 

\begin{proof} 
The conclusion of the first part is known (see \cite{XCY24} for example).  
Then the desired conclusions of the remaining two parts follow from Lemma \ref{lem-AKS}.  
\end{proof} 

For any fixed integer $\ell \geq 2$,  define 
$s_2(\ell,\lambda)$ to be the smallest positive integer such that $(2\ell-1)\lambda \leq 2^{\lceil s_2(\ell,\lambda)/2 \rceil}$.  
We consider the BCH code $\C_{(2,n, 2\ell, 0)}$.  If $s \geq s_2(\ell,\lambda)$, by definition and Lemma 
\ref{lem-secfinal},  
\begin{eqnarray*}
\dim(\C_{(2,n, 2\ell, 0)}) &=& n-\sum_{i=1}^{\ell-1} \deg(\m_{\beta^{2i-1}}(x))-1 \\ 
&=& n-\sum_{i=1}^{\ell-1} |C_{2i-1}| -1 \\ 
&=& n-(\ell-1)s-1.
\end{eqnarray*}

By the BCH bound we have 
$d(\C_{(2,n, 2\ell, 0)})\geq 2\ell$.  
Note that 
$$ 
\ell! \left(\sum_{i=0}^\ell \binom{n}{i} -2(\lambda n+1)^{\ell-1} \right)
=n^\ell +\sum_{i=0}^{\ell-1} a_i(\ell, \lambda) n^i, 
$$ 
where these $a_i(\ell, \lambda)$ are integers depending on $\ell$ and $\lambda$ only.  
Define $$a(\ell, \lambda) =\max_{0 \leq i \leq \ell-1}\{|a_i(\ell, \lambda) |\}$$ and 
 $s_1(\ell, \lambda)$ to be the smallest positive integer such that $1+\ell \lambda  a(\ell, \lambda) < 2^{\lceil s_1(\ell,\lambda)/2 \rceil}$.   
If $s \geq s_1(\ell,\lambda)$, we have 
\begin{eqnarray*}
 n^\ell +\sum_{i=0}^{\ell-1} a_i(\ell, \lambda)  n^i 
 \geq    n^\ell -\sum_{i=0}^{\ell-1}| a_i(\ell, \lambda) |n^i 
\geq    n^\ell -  a(\ell, \lambda) \ell n^{\ell-1} 
>0. 
\end{eqnarray*}
Define $s(\ell, \lambda) =\max\{s_1(\ell, \lambda), s_2(\ell,\lambda)\}$.  
It then follows from the sphere packing bound that $d(\C_{(2,2^s-1, 2\ell, 0)})\leq 2\ell$ if $s\geq s(\ell, \lambda)$. 
Consequently, $d(\C_{(2,2^s-1, 2\ell, 0)}) = 2\ell$ if $s\geq s(\ell, \lambda)$. 
We  have now  proved the following theorem. 

\begin{theorem}\label{thm-infinite}
Let notation be the same as before.  If $s\geq s(\ell, \lambda)$, then the BCH code $\C_{(2,n, 2\ell, 0)}$ has parameters $[n,  n-1-(\ell-1)s, 2\ell]$ and is distance-optimal with respect to the sphere packing bound. 
\end{theorem}

\begin{remark} 
Theorem \ref{thm-infinite} is a generalisation of some known results in the literature.  For example, 
when $\lambda=1$ and $\ell \in \{1,2,3,4 \}$, the parameters of the codes  $\C_{(2,n, 2\ell, 0)}$ are known in 
the literature (see \cite{Dingbook18} for example), although the distance-optimality with respect to 
the sphere packing bound may not have been mentioned in the literature.  
When $\ell =2$ and $\ell=4$, the code $\C_{(2,n, 2\ell, 0)}$ was known to be distance-optimal \cite{Xieetal}.  
When $\ell =3$, the code $\C_{(2,n, 2\ell, 0)}$ was known to be distance-optimal \cite{ChenWu}
\end{remark} 

\begin{remark}
For each constant divisor $\lambda$ of $2^s-1$ and each $\ell \geq 2$, there is a constant 
$s(\ell, \lambda)$ such that $\{ \C_{(2,n, 2\ell, 0)}: s \geq  s(\ell, \lambda)\}$ is an infinite family of 
distant-optimal binary cyclic codes with respect to the sphere packing bound, where each code 
$\C_{(2,n, 2\ell, 0)}$  has parameters 
$[n,  n-1-(\ell-1)s, 2\ell]$.  This shows that Theorem \ref{thm-infinite} has given infinitely many families of 
distant-optimal binary cyclic codes with respect to the sphere packing bound.  Hence,  
Theorem \ref{thm-infinite} has made a 75-year breakthrough in coding theory in the sense that 
it gives the first infinite family of distance-optimal binary cyclic codes with arbitrarily large minimum distance 
with respect to the sphere packing bound. 
\end{remark}

Table \ref{tab-one} provides some experimental data on the threshold value $s(\ell, \lambda)$ for code $\C_{(2,n, 2\ell, 0)}$ being distance-optimal  with respect to the sphere packing bound when $s \geq s(\ell, \lambda)$.  

\begin{table}[htbp] 
\centering 
\caption{The threshold value $s(\ell, \lambda)$ for $ \C_{(2,2^s-1, 2\ell, 0)}$ being distance-optimal}\label{tab-one}
\begin{tabular}{|r|r|r|}
\hline 
$\lambda$ & $\ell$ & $s(\ell, \lambda)$  \\ \hline 
1 & 2 &    3  \\  \hline 
1 &3  &    4  \\  \hline 
1 &4   &    6  \\  \hline 
1 &5 &   8   \\  \hline 
1 &6  & 11     \\  \hline 
1 & 7   & 14     \\   \hline  
1 &8 &  17    \\  \hline 
1 & 9 & 20     \\  \hline 
1 & 10   &  23    \\   \hline  
\end{tabular}
\end{table}

\section{Open problems}\label{sec-open} 

For the three types of binary BCH codes $\C_{(2,n,  \delta,   b)}$, where 
\begin{itemize}
\item $n=(2^{2s}+1)(2^s-1)$ for positive integer $s$, or 
\item $n=2^{2s}+2^s+1$ for positive integer $s$, or 
\item $n=(2^{s}-1)/\lambda$ for positive integer $s$ and $\lambda >1$ being a constant divisor of 
$2^s-1$, 
\end{itemize} 
we have the following open problems.  

\begin{problem} 
Determine the automorphism groups of code $\C_{(2,n,  \delta,   b)}$ and code $\overline{\C_{(2,n,  \delta,   b)}}$. 
\end{problem} 

\begin{problem} 
Does $\overline{\C_{(2,n,  \delta,   1)}}$ support $2$-designs for certain $\delta$ \cite{Dingbook18}? 
\end{problem} 

\begin{problem} 
Determine $d(\C_{(2,n,  \delta,   b)})$ for $\delta \geq 7$ and $d(\overline{\C_{(2,n,  \delta,   b)}})$ for $\delta \geq 8$.  
\end{problem}  

\begin{problem} 
Are there $\delta$ and $b$ such that $\dim(\C_{(2,n,  \delta,   b)}) \geq (n-1)/2$ and 
$d(\C_{(2,n,  \delta,   b)}) \geq \sqrt{n}/2$?   
\end{problem}  

\begin{problem} 
Determine $d(\C_{(2,n,  \delta,   b)}^\perp)$ and $d(\overline{\C_{(2,n,  \delta,   b)}}^\perp)$ 
for $\delta \geq 7.$ 
\end{problem}  

\begin{problem}
Determine the covering radii of the distance-optimal binary codes presented in this paper. 
\end{problem}

It would be interesting to settle some of the open problems under certain conditions.

\section{Summary and concluding remarks}\label{sec-summary}

The following infinite families of distance-optimal binary linear codes were presented in this paper. 

\begin{itemize}
\item The infinite family of the binary codes $\overline{\C_{(2,n,  3,   1)}}$  with parameters $[n+1, n-4s, 4]$,  where $n=(2^{2s}+1)(2^s-1)$ and $s \geq 2$ (see Theorem \ref{thm-one2}).

\item The infinite family of the binary codes $\C_{(2,n, 6,0)}$  with parameters $[n+1, n-8s-1, 6]$,  where $n=(2^{2s}+1)(2^s-1)$ and $s \geq 8$ (see Theorem \ref{thm-one2-add}).

\item The infinite family of the binary codes $ \overline{\C_{(2,n,  5,   1)}}$ with parameters 
$[n+1, n-8s, 6]$,  where $n=(2^{2s}+1)(2^s-1)$ and $s \geq 4$ (see Theorem \ref{thm-s31}).

\item The infinite family of the binary codes $\overline{\C_{(2,n,3,1)}}$ with parameters $[n+1, n-3s,4]$, where $n=2^{2s}+2^s+1$ and $s\geq 2$ (see Theorem \ref{thm-two1}).  

\item The infinite family of the binary codes $\overline{\C_{(2,n,  3,   1)}}$  with parameters $[n+1, n-2s, 4]$, where $n=(4^s-1)/3$ and 
$s \geq 3$ (see Theorem \ref{thm-five1}).  

\item The infinite family of the binary codes $ \overline{\C_{(2,n,  5,   1)}}$ with parameters 
$[n+1, n-4s, 6]$,  where $n=(4^s-1)/3$ and 
$s \geq 3$ (see Theorem \ref{thm-s61}).  

\item Infinitely many families of the binary cyclic codes $\C_{(2,(2^s-1)/\lambda, 2\ell, 0)}$ with parameters $[(2^s-1)/\lambda,  (2^s-1)/\lambda-1-(\ell-1)s, 2\ell]$, where $\ell \geq 2$ and $\lambda$ is a constant divisor of $2^s-1$ (see Theorem \ref{thm-infinite}). 
\end{itemize}
In addition, two infinite families of five-weight binary codes were presented in this paper.  
This paper has made a 75-year breakthrough in coding theory in the sense that 
it presents the first infinite family of distance-optimal linear codes with arbitrarily large minimum distance 
with respect to the sphere packing bound. 

The reader is informed that the idea for constructing distance-optimal binary cyclic codes with 
respect to the sphere packing bound in  Section \ref{sec-newnew} does not work for the code $\C_{(2,n, 2\ell, 0)}$ for $n$ being 
of the following firms: 
\begin{itemize}
\item $n=(2^{2s}+1)(2^s-1)=\frac{2^{4s}-1}{2^s+1}$ for positive integer $s$ and $\ell \geq 4$ (it works only for $\ell=2$ and $\ell =3$). 
\item $n=2^{2s}+2^s+1=\frac{2^{3s}-1}{2^s-1}$ for positive integer $s$ and $\ell \geq 3$ (it works only for $\ell=2$). 
\end{itemize} 
In the two cases $\lambda=2^s+1$ and $\lambda=2^s-1$, the value of $\lambda$ depends on $s$ and 
is not a constant divisor of $2^{4s}-1$ and $2^{3s}-1$.  When $n=\frac{2^{4s}-1}{2^s+1}$ and $\ell \geq 4$, and when $n=\frac{2^{3s}-1}{2^s-1}$ and $\ell \geq 3$,  the sphere packing bound cannot be used to prove that $d(\C_{(2,n,2\ell,0}))\leq 2\ell$ for $s$ being large enough.  
This justifies that we need to treat the BCH codes of length $n$ of different forms separately.

\end{document}